\documentclass[11pt,fleqn]{article}

\usepackage{amssymb,amsmath}
\usepackage{natbib}
\usepackage{vmargin,setspace,verbatim}
\usepackage{graphicx,graphics,color,tikz,tikz-network}
\usepackage{microtype}

\usepackage{libertine}
\usepackage[libertine]{newtxmath}

\definecolor{darkgreen}{rgb}{0,0.2,0}
\definecolor{darkred}{rgb}{0.3,0,0}

\usepackage[colorlinks=true, pdfstartview=FitV, linkcolor=darkgreen, citecolor=darkred, urlcolor=black]{hyperref}

\newcounter{llst}
\newenvironment{abet}{\begin{list}{\rm (\alph{llst})}{\usecounter{llst}
\setlength{\itemindent}{0em} \setlength{\leftmargin}{3em}
\setlength{\labelwidth}{2em} \setlength{\labelsep}{1em}}}{\end{list}}

\newcounter{llist}
\newenvironment{numm}{\begin{list}{\rm (\roman{llist})}{\usecounter{llist}
\setlength{\itemindent}{0em} \setlength{\leftmargin}{3.5em}
\setlength{\labelwidth}{2.5em} \setlength{\labelsep}{1em}}}{\end{list}}

\newtheorem{theorem}{Theorem}[section]

\newtheorem{corollary}[theorem]{Corollary}

\newtheorem{definition}[theorem]{Definition}
\newtheorem{expl}[theorem]{Example}

\newtheorem{lemma}[theorem]{Lemma}

\newtheorem{proposition}[theorem]{Proposition}

\newtheorem{dscrpt}[theorem]{Description}

\newcounter{axiomatiser}
\newcounter{myclaimcount}
\newenvironment{proof}[1][Proof]{\noindent \textbf{#1.} }{\hfill
\rule{0.5em}{0.5em}}

\newenvironment{example}{\begin{expl} \rm}{\hfill $\blacklozenge$
\end{expl}}

\setpapersize{A4}
\setmarginsrb{80pt}{40pt}{80pt}{60pt}{12pt}{10pt}{12pt}{30pt}

\begin{document}


\title{\textbf{Construction of Compromise Values \\ for Cooperative Games}\thanks{We thank Lina Mallozzi for comments, suggestions and help with progressing on the ideas set out in this paper.}}

\author{Robert P.~Gilles\thanks{\textbf{Corresponding author:} Department of Economics, The Queen's University of Belfast, Riddel Hall, 185 Stranmillis Road, Belfast, BT9~5EE, UK. \textsf{Email: r.gilles@qub.ac.uk}} \and Ren\'e van den Brink\thanks{ VU University, Department of Economics, and Tinbergen Institute, De Boelelaan 1105, 1081 HV~Amsterdam, the Netherlands. \textsf{Email: j.r.vanden.brink@vu.nl}} }

\date{July 2026}

\maketitle

\begin{abstract}
\singlespace
\noindent
We explore a broad class of values for cooperative games in characteristic function form, known as \emph{compromise values\/}. These values efficiently allocate payoffs by linearly combining well-specified upper and lower bounds on payoffs. We identify subclasses of games that admit non-trivial efficient allocations within the considered bounds, which we call \emph{bound-balanced games}. Subsequently, we define the associated compromise value. We also provide an axiomatisation of this class of compromise values using variants of the minimal rights property and restricted proportionality.

We introduce two construction methods for properly devised compromise values. Under mild conditions, one can use either a lower or an upper bound to construct a well-defined compromise value.

Based on these methods, we construct and axiomatise various well-known and new compromise values, including the $\tau$-, the $\chi$-, the Egalitarian, the Gately, the CIS-, the PANSC-, the EANSC-, and the new KM-values. We conclude that this approach establishes a common foundation for a wide range of different values.
\end{abstract}

\begin{description}
\singlespace
\item[Keywords:] Cooperative TU-game; compromise value; $\tau$-value; CIS value; Gately value; KM-value; axiomatisation; bound balanced games.

\item[JEL classification:] C71
\end{description}

\thispagestyle{empty}

\newpage

\setcounter{page}{1} \pagenumbering{arabic}

\section{Introduction}

In the theory of cooperative games with transferable utility (TU-games), the analysis of values has resulted in a large, prominent literature. A value on a particular class of TU-games assigns to every game in this class an efficient allocation of the proceeds of the grand coalition to the individual members of the population. Most prominent are the Shapley value \citep{Shapley1953} and its variations, which are firmly rooted in well-accepted axiomatic characterisations. Other values are founded on selecting Core imputations, if the Core \citep{Edgeworth1881,Gillies1959} is non-empty, such as the Nucleolus \citep{Schmeidler1969}. In this paper, we investigate an alternative class of values, known as \emph{compromise values}.  

A compromise value assigns the efficient balance of a given upper and lower bound to any TU-game for which the stated bounds are well-defined. Specifically, it assigns to every game on an appropriate subclass of TU-games the unique efficient allocation on the line segment between these two bounds.  The main precedent of such a compromise value is the $\tau$-value \citep{Tijs1981}, based on balancing the marginal contributions as an upper bound and the minimal rights payoffs as a lower bound. The $\tau$-value is well-defined on the class of semi-balanced games \citep{Tijs1981}. Further analysis of the $\tau$-value as a compromise value and its axiomatisation on the class of quasi-balanced games has been developed in \citet{TijsDriessen1985,TijsDriessen1987,Tijs1987,CalvoTijs1995,Chi1996,Compromise2001} and \citet{Compromise2005}.

Although we restrict our definitions and analysis of compromise values to the realm of cooperative games in characteristic function form with transferable utility, \citet{Otten1990,Compromise1992} as well as \citet{TijsOtten1993} extend the concept of compromise values to other realms such as non-transferable utility games and bargaining problems. 

We investigate a broad class of compromise values based on so-called \emph{bound pairs} comprising two functionals on the space of all TU-games. These functionals must satisfy three regularity properties that are naturally satisfied by those representing upper and lower bounds on allocations in a TU-game. This general approach encompasses the construction method introduced in \citet{Compromise2000note} which imposes a more restrictive covariance property on the upper bound functional.

Our approach allows the identification of a rather broad class of compromise values including well-known values such as the $\tau$-, $\chi$-, Egalitarian Value, Gately, Centre of the Imputation Set (CIS), Proportional Allocation of Non-Separable Contributions (PANSC), and Equal Allocation of Non-Separable Contributions (EANSC) values. We also introduce a new compromise value, the Kikuta-Milnor (KM) value based on a lower bound function introduced by \citet{Kikuta1980} and an upper bound functional considered by \citet{Milnor1952}.

Compromise values for a given bound pair are only well-defined on the subclass of bound-balanced TU-games corresponding to the TU-games that admit efficient allocations between the given upper and lower bounds. With reference to the $\tau$-value introduced by \citet{Tijs1981}, this subclass is that of quasi- and semi-balanced games. We show that the KM-bound pair---based on the bound functionals introduced by \citet{Kikuta1980} and \citet{Milnor1952}---admits the complete space of TU-games as the corresponding KM-bound-balanced games. The other compromise values considered here usually refer to strict subclasses of the space of all TU-games. 

\paragraph{Constructing bound pairs and their associated compromise values.}

We propose two methods to construct compromise values based on a single, fixed functional. First, we consider constructing compromise values using a lower bound functional that has a regularity property. The allocation functional that gives each player the remaining value after paying others their lower bound forms a bound pair with the lower bound functional. Therefore, the resulting compromise value depends entirely on the chosen lower bound. We show that the Egalitarian Value, the CIS-value, and the EANSC value are compromise values that can be constructed this way using the right lower bound functionals.

We demonstrate that these lower bound-based compromise values can be axiomatised by a variant of the minimal rights property and an egalitarian allocation property. The latter imposes the assignment of an equal share to all players for lower bound normalised games. Our variant of the egalitarian allocation property is more restrictive than the restricted proportionality property of \citet{Tijs1987}.

Second, we consider the construction of compromise values based on a given upper bound functional only. This construction method follows the framework set out by \citet{Tijs1981} for the $\tau$-value. This method identifies a minimal rights payoff vector based on the given upper bound on the payoffs in the TU-game. We show that if the upper bound functional is translation covariant, the constructed minimal rights payoff vector indeed defines an appropriate lower bound functional to form a bound pair with the given upper bound functional. The considered covariance property on the upper bound functional is less stringent than the covariance property introduced by \citet{Compromise2000note}. 

This paper considers a similar construction method for a corresponding lower bound, which is the minimal rights payoff vector for the given upper bound. Therefore, our construction method generalizes the construction method introduced by \citet{Compromise2000note}. We demonstrate that the $\tau$-value and the CIS-value can be constructed using this method. This implies that the CIS-value can be constructed from its characteristic lower bound functional as well as its characteristic upper bound functional. This highlights the singular nature of the CIS-value as a value that can be constructed from both a lower bound as well as an upper bound.

\paragraph{Axiomatisations}

We also consider axiomatisations of compromise values. The compromise value with respect to some bound pair is the unique allocation rule that---besides the standard efficiency property---satisfies two main properties that both are modifications of corresponding axioms in \citet{Tijs1987}: A form of covariance, known as the \emph{minimal rights} property, and the \emph{restricted proportionality} property. We show this characterisation to hold on the class of \emph{all} bound pairs.  

This type of axiomatisation is well-established in the literature for certain (classes of) compromise values, such as the seminal contribution by \citet{Tijs1981} on the $\tau$-value and \citet{Compromise2000note} for the upper bound constructed class of compromise values, replacing the minimal rights property with a covariance property. We complement this with a comparable characterisation for the class of lower bound constructed compromise values based on replacing the restricted proportionality with an egalitarian division property.

\paragraph{Motivation and applications.}

In the literature, different compromise values have been applied to a variety of problems. The $\tau$-value has been applied to delivery games \citep{Hamers1997} and to information sharing in Data Envelopment Analysis \citep{Lozano2012}, and its extension to games with a coalition structure \citep{Tau2003} underlies further applications. The Gately value was initially applied to sharing the gains from regional cooperation in electric power planning \citep{Gately1974} and, more recently, to measuring centrality in networks \citep{GillesMallozzi2024}. The PANSC- and EANSC values arise in cost-sharing problems; in particular, the Separable Cost Remaining Benefits and Alternative Costs Avoided methods coincide under general conditions with the PANSC value \citep{PANSC2023}.

From these applications it can be argued that no single compromise value is better than the others, since which value is most appropriate depends on the application at hand. Therefore, instead of focusing on one compromise value, this paper develops a more general methodology for building values from lower and upper bounds. In this way, the results in this paper (i) give insight into the differences between the compromise values that are most appropriate for different applications, and (ii) show that these different compromise values are nonetheless built on a common methodology for determining these bounds.

\paragraph{Structure of the paper.}

Section~2 collects the necessary preliminaries and reviews the $\tau$-value. Section~3 introduces compromise values for general bound pairs, extends the axiomatisation of \citet{Tijs1987} to this class, and illustrates it with the PANSC-, Gately and KM-values. Section~4 constructs compromise values from regular lower bounds and shows that the Egalitarian Value and the CIS-value arise in this way; Section~5 does the same for translation covariant upper bounds, recovering the $\tau$- and $\chi$-values as well as the CIS-value. Section~6 concludes by investigating the EANSC value, which admits both constructions.

\section{Preliminaries: Cooperative games and values}

We first discuss the foundational concepts of cooperative games and solution concepts. Let $N = \{ 1, \ldots , n \}$ be an arbitrary finite set of players and let $2^N = \{ S \mid S \subseteq N \}$ be the corresponding set of all (player) coalitions in $N$. For ease of notation we usually refer to the singleton $\{ i \}$ simply as $i$. Furthermore, we use the simplified notation $S-i = S \setminus \{ i \}$ for any $S \in 2^N$ and $i \in S$ as well as $S + i = S \cup \{i\}$ for any $S \in 2^N$ and $i \in N \setminus S$. 

A cooperative game with transferable utility---shortly referred to as a \emph{cooperative game}, or simply as a \emph{game}---on $N$ is a function $v \colon 2^N \to \mathbb R$ such that $v ( \varnothing ) =0$. A game assigns to every coalition a value, which we call the coalition's {\em worth\/}, being the value that this coalition can generate through the cooperation of its members.

The class of all cooperative games in the player set $N$ is denoted by
\begin{equation}
	\mathbb V^N = \left\{ v \, \left| \, v \colon 2^N \to \mathbb R \mbox{ such that } v ( \varnothing ) =0 \right. \right\}
\end{equation}
The class $\mathbb V^N$ forms an $(2^n-1)$-dimensional linear real vector space.

The \emph{dual} of a game $v \in \mathbb V^N$ is the game $v^* \in \mathbb V^N$ defined by $v^* (S) = v(N) - v(N \setminus S)$ for all $S \subseteq N$.

We use the natural ordering of the Euclidean space $\mathbb V^N$ to compare different cooperative games. In particular, we denote $v \leqslant w$ if and only if $v(S) \leqslant w(S)$ for all $S \neq \varnothing$; $v<w$ if and only if $v \leqslant w$ and $v \neq w$; and, finally, $v \ll w$ if and only if $v(S) < w(S)$ for all $S \neq \varnothing$.\footnote{We use the same notational convention to compare vectors in arbitrary Euclidean spaces.}

One can consider several bases of this linear vector space. For the study of compromise values the so-called standard base is the most useful. Formally, the \emph{standard base} is given by  the finite collection of games $\{ b_S \mid S \subseteq N \} \subset \mathbb V^N$, for every coalition $S \subseteq N$ defined by\footnote{As such the standard base corresponds to the standard base of unit vectors in the Euclidean vector space $\mathbb V^N$.}
\[
b_S (T) = \left\{
\begin{array}{ll}
	1 & \mbox{if } T=S \\
	0 & \mbox{if } T \neq S
\end{array}
\right.
\]
For any vector of payoffs $x \in \mathbb R^N$ we denote $x(S) = \sum_{i \in S} x_i$ for any coalition of players $S \in 2^N$ as the total assigned payoff to a certain coalition of players. As such, $x \in \mathbb R^N$ defines a corresponding trivial additive game $x \in \mathbb V^N$.  Hence, the game $x \in \mathbb V^N$ assigns to every coalition the sum of the individual contributions of its members; there are no cooperative effects from bringing players together in that coalition.

For every player $i \in N$, let $v_i = v ( \{ i \} )$ be their individual worth in the game $v$. We refer to the game $v$ as being \emph{zero-normalised} if $v_i =0$ for all $i \in N$. We denote the collection of all zero-normalised games by $\mathbb V^N_0 \subset \mathbb V^N$. 

For any cooperative game $v \in \mathbb V^N$ and vector $x \in \mathbb R^N$, we denote by $v+x \in \mathbb V^N$ the cooperative game defined by $(v+x) (S) = v(S) + x(S) = v(S) + \sum_{i \in S} x_i$ for every coalition $S \in 2^N$. Further, for any cooperative game $v \in \mathbb V^N$, denoting $\underline{\nu} (v) = (v_1, \ldots v_n) \in \mathbb R^N$ the vector of individual worths, we refer to the game $v- \underline{\nu} (v) \in \mathbb V^N_0$, as the \emph{zero-normalisation} of the game $v$. 

\subsection{Game properties and values}

We explore well-known properties of games that are required in our analysis. We refer to a game $v \in \mathbb V^N$ as (i) \emph{monotonic} if $v(S) \leqslant v(T)$ for all $S,T \subseteq N$ with $S \subseteq T$ and (ii) \emph{superadditive} if for all $S,T \subseteq N$ with $S \cap T = \varnothing$ it holds that $v(S \cup T) \geqslant v(S) + v(T)$.

\medskip\noindent
The \emph{marginal contribution}---also known as the ``utopia'' value \citep{Tijs1981,Tijs2008}---of an individual player $i \in N$ in the game $v \in \mathbb V^N$ is defined by their marginal or ``separable'' contribution \citep{Moulin1985} to the grand coalition in this game, i.e.,
\begin{equation} \label{eq:MCdef}
	M_i (v) = v(N) - v(N-i) .
\end{equation}
We call a cooperative game $v \in \mathbb V^N$ \emph{essential} if it holds that
\begin{equation} \label{eq:essential}
	\sum_{j \in N} v_j \leqslant v(N) \leqslant \sum_{j \in N} M_j (v)
\end{equation}
The class of essential games is denoted as $\mathbb V^N_E \subset \mathbb V^N$, which forms a linear subspace of $\mathbb V^N$.\footnote{We remark that in the literature, the term ``weakly essential'' or ``essential'' frequently pertains solely to the first inequality in (\ref{eq:essential}).}

\medskip\noindent
A game $v \in \mathbb V^N$ is called \emph{convex} if for all coalitions $S,T \in 2^N$ it holds that $v(S \cup T) + v (S \cap T) \geqslant v(S) + v(T)$ \citep{Shapley1971}. \citet{Driessen1985} pointed out that a game $v \in \mathbb V^N$ is convex if and only if for all $i \in S \subset T \colon v(S) - v(S-i) \leqslant v(T) - v(T-i)$. We denote the subclass of convex games by $\mathbb V^N_C \subset \mathbb V^N$.

\medskip\noindent
An \emph{allocation}---also known as a ``pre-imputation''---in the game $v \in \mathbb V^N$ is any point $x \in \mathbb R^N$ such that $x(N)=v(N)$.  We denote the class of all allocations for the game $v \in \mathbb V^N$ by $\mathcal A (v) = \{ x \in \mathbb R^N \mid x(N) = v(N) \} \neq \varnothing$. We emphasise that allocations can assign positive as well as negative payoffs to individual players in a game. 

An \emph{imputation} in the game $v \in \mathbb V^N$ is an allocation $x \in \mathcal A(v)$ that is individually rational in the sense that $x_i \geqslant v_i$ for every player $i \in N$. The corresponding imputation set of $v \in \mathbb V^N$ is given by $\mathcal I (v) = \{ x \in \mathcal A(v) \mid x_i \geqslant v_i$ for all $i \in N \}$. We remark that $\mathcal I(v)$ is a polytope in $\mathcal A (v)$ for any cooperative game $v \in \mathbb V^N$. Furthermore, $\mathcal I (v) \neq \varnothing$ if and only if $v(N) \geqslant \sum_{i \in N} v_i$. In particular, this is the case for essential games. 

A value on a class of games is a function that assigns an allocation to every game in this class.
\begin{definition}
	Let $\mathbb V \subseteq \mathbb V^N$ be some subclass of cooperative games on player set $N$. A \textbf{value} on $\mathbb V$ is a function $f \colon \mathbb V \to \mathbb R^N$ such that $f(v) \in \mathcal A (v)$ for every $v \in \mathbb V$.
\end{definition}
Note that, since every allocation in a game $v \in \mathbb V^N$ distributes the total worth $v(N)$ that is created in that game, a value is defined to be \emph{efficient} in the sense that $\sum_{i \in N} f_i (v) = v(N)$ for all $v \in \mathbb V$.

We remark that a value assigns an allocation to every game in the selected subclass, but not necessarily an imputation. We call a value $f \colon \mathbb V \to \mathbb R^N$ \emph{individually rational} if for every $v \in \mathbb V \colon f(v) \in \mathcal I (v)$. Hence, an individually rational value on a subclass of games explicitly assigns an imputation to every game in that subclass. We remark that many of the compromise values considered in this paper are individually rational. 

The {\em dual value} $f^* \colon \mathbb V^* \to \mathbb R^N$ of a value $f \colon \mathbb V \to \mathbb R^N$ is obtained by applying the value $f$ to the dual game $v^*$: $f^*(v)=f(v^*)$ for all $v \in \mathbb V$. Consider a class of games $\mathbb V \subseteq \mathbb V^N$ satisfying that $v \in \mathbb V$ if and only if $v^* \in \mathbb V$. A value $f$ is called {\em self-dual} on $\mathbb V$ if and only if $f(v)=f^*(v)$ for all $v \in \mathbb V$.

Finally, for some class of games $\mathbb V \subseteq \mathbb V^N$, a value $f \colon \mathbb V \to \mathbb R^N$ is \emph{translation covariant} if $f(v+x) = f(v) + x$ for any game $v \in \mathbb V$ and vector $x \in \mathbb R^N$.\footnote{This property is referred to as S-equivalence by \citet{Tijs1981}.}

\subsection{The $\tau$-value}

\citet{Tijs1981,Tijs1987} seminally introduced the quintessential compromise value, the $\tau$-value. This value is explicitly constructed as an allocation that is the balance of the marginal contribution vector $M(v)$ defined by (\ref{eq:MCdef}) as an upper bound and a lower bound that is derived from $M(v)$. In particular, \citet{Tijs1981} selected the lower bound that assigns to every player the maximal `surplus' over all coalitions it is a member of where this surplus is what is left of the worth of that coalition after all other players in the coalition received their marginal contribution. Thus, the lower bound is given by
\begin{equation} \label{eq:MinimalRights}
	m_i (v) = \max_{S \subseteq N \colon i\in S} R_i (S,v) \qquad \mbox{where } R_i (S,v) = v(S) - \sum_{j \in S - i} M_j (v)
\end{equation}
\citet{Tijs1987} refers to $m(v)$ as the ``minimal rights'' vector for the game $v$.

\paragraph{Semi-balanced games}

Following \citet{Tijs1981}, we call a game $v \in \mathbb V^N$\emph{semi-balanced} if for every coalition $S \subseteq N$ it holds that
\begin{equation}
	v(S) + \sum_{j \in S} v(N-j) \leqslant |S| \, v(N)
\end{equation}
The class of semi-balanced games is denoted by $\mathbb V^N_S \subset \mathbb V^N$.

The upper bound $M$ and the lower bound $m$ defined above form a bound pair $(m,M)$ on $\mathbb V^N_S$ in the sense that for every semi-balanced game $v \in \mathbb V^N_S \colon m(v) \leqslant M(v)$, $m(v-m(v))=0$ and $M(v-m(v)) = M(v) - m(v)$.

\citet{Tijs1981} already showed that there is a close relationship between semi-balanced games and the introduced bound pair $(m,M)$ in the sense that
\[
\mathbb V^N_S = \left\{ v \in \mathbb V^N \, \left| \, \sum_{i \in N} m_i(v) \leqslant v(N) \leqslant \sum_{i \in N} M_i(v) \, \right. \right\}
\]
\citet{Tijs1981} defined the $\tau$-value as the value on the class of semi-balanced games $\mathbb V^N_S$ that assigns to every game in this class the allocation that balances $m$ and $M$. Hence, the $\tau$-value of a semi-balanced game $v \in \mathbb V^N_S$ is  given by $\tau (v) = \lambda_v m(v) + (1- \lambda_v ) M(v)$, where $\lambda_v \in [0,1]$ is determined such that $\sum_{i \in N} \tau_i (v) = v(N)$. It can be computed that for every $v \in \mathbb V^N_S$ and $i \in N \colon$
\begin{equation}
	\tau_i (v) = m_i (v) + \frac{M_i (v) - m_i (v)}{\sum_{j \in N} \left( \, M_j (v) - m_j (v) \, \right)} \left( \, v(N) - \sum_{j \in N} m_j (v) \, \right)
\end{equation}
In this paper, we set out to develop a generalisation of the notion introduced by \citet{Tijs1981} to general compromise values based on a large class of bound pairs. This is discussed in the next section.

As an application, \citet{NunezRafels2002} show that for two-sided assignment games the $\tau$-value coincides with the fair division point \citep{Thompson1994}, the average of the buyer- and seller-optimal Core payoff vectors. 

\section{Compromise Values}

We aim to generalise the construction method on which the $\tau$-value is founded to arbitrary bound pairs on well-defined subclasses of cooperative games. In particular, we aim to find bound pairs for which this method results in meaningful compromise values. The determinants of this construction are the chosen lower and upper bounds, which are determined by well-chosen functionals on the relevant class of cooperative games. 

\subsection{Bound pairs}

We first give a definition that provides necessary requirements for a constructed bound pair consisting of a lower- and upper functional to be interpreted as bounds for the allocations.
\begin{definition} \label{def:BoundPair}
	Let $\mathbb V \subseteq \mathbb V^N$ be some class of games on player set $N$. A pair of functions $(\mu , \eta ) \colon \mathbb V \to \mathbb R^N \times \mathbb R^N$ is a \textbf{bound pair} on $\mathbb V$ if these functions satisfy the following properties:
	\begin{numm}
		\item For every $v \in \mathbb V \colon \mu (v) \leqslant \eta (v)$, and
		\item For every $v \in \mathbb V \colon v - \mu (v) \in \mathbb V$ and the following two properties hold:
		\begin{abet}
			\item For every $v \in \mathbb V \colon \mu (v - \mu (v)) =0$, and
			\item For every $v \in \mathbb V \colon \eta (v - \mu (v)) = \eta (v) - \mu (v)$.
		\end{abet}
	\end{numm}
	A bound pair $(\mu , \eta )$ is \textbf{proper} on $\mathbb V$ if there exists at least one game $v \in \mathbb V$ with $\mu (v) < \eta (v)$ and it is \textbf{strict} on $\mathbb V$ if there exists at least one game $v \in \mathbb V$ with $\mu (v) \ll \eta (v)$.
\end{definition}
The function $\mu$ can be interpreted as a lower bound for every game in the subclass $\mathbb V \subseteq \mathbb V^N$, while the function $\eta$ assigns an upper bound. The three properties introduced in the definition of a bound pair ensure that these functions behave as one might expect from lower and upper bounds. 

Property (i) imposes the natural order that an upper bound is at least as large as the lower bound. Properties (ii-a) and (ii-b) are weak versions of covariance. These properties restrict the choice of bound pairs considerably. As a consequence the corresponding compromise value is well-defined only for cooperative games for which such a combination of a lower- and upper functionals is indeed a bound pair.

The intuition of property (ii-a) is based on the interpretation of $\mu$ as a proper lower bound. For any $v \in \mathbb V$, the derived game $v - \mu (v) \in \mathbb V^N$ is defined by $(v - \mu (v)) (S) = v(S) - \sum_{j \in S} \mu_j (v)$ for coalition $S \subseteq N$. Hence, $v - \mu (v)$ assigns to a coalition the originally generated worth minus the payoff assigned by the lower bound to all members of that coalition. It is natural to require that the natural lower bound of $0 \in \mathbb R^N$ is assigned to this reduced game. 
 \begin{example} \label{expl:lowerZero}
Consider any $m^0 \in \mathbb R^N$ and the constant function $m^0 (v) = m^0$ for all $v \in \mathbb V^N$. Since $m^0 (v - m^0(v)) = m^0$, this constant function violates property (ii-a) unless $m^0 =0$. The zero lower bound $\mu^0$ with $\mu^0(v)=0$, by contrast, is a genuine lower bound: for \emph{any} upper bound $\eta$ with $\eta (v) \geqslant 0$ it forms a proper bound pair $(\mu^0 , \eta )$ on the whole space $\mathbb V^N$, even for non-linear $\eta$ (as for the PANSC value in Section 3.3.1).
\end{example}
Not every natural pair of bounds combines into a proper bound pair, however.
\begin{example}
The individual-worth lower bound $\underline{\nu} (v) = (v_1, \ldots ,v_n )$ and the marginal-contribution upper bound $M(v) = ( M_1(v), \ldots , M_n (v) \, )$ form a proper bound pair $(\underline{\nu} ,M)$ on the class of essential games $\mathbb V^N_E$. By contrast, pairing $\underline{\nu}$ with the constant upper bound $\eta^0 (v) = ( v(N), \ldots ,v(N) \, )$ fails property (ii-b) for every non zero-normalised game, since then
\[
\eta^0 ( v - \underline{\nu} (v) \, )_i = v(N) - \sum_{j \in N} v_j \;\neq\; v(N) - v_i = \eta^0 (v)_i - \underline{\nu}_i (v) .
\]
Thus even natural lower and upper bounds cannot always be combined into bound pairs.
\end{example}

\paragraph{Linear bound pairs}
	
Some of the bound pairs considered in this paper are linear functionals on the space of all games $\mathbb V^N$. This simply means that these bounds can be written as weighted sums of coalitional worths assigned in the particular game $v \in \mathbb V^N$ under consideration.

Formally, a function $f \colon \mathbb V^N \to \mathbb R$ is \emph{linear} if for every game $v \in \mathbb V^N$ and every player $i \in N \colon f_i (v) = \alpha^i \cdot v$ for some $\alpha^i \in \mathbb V^N$ and the operator ``$\cdot$'' refers to the inner product on $\mathbb V^N$ as a Euclidean vector space.\footnote{Hence, we can also write $f_i (v) = \sum_{S \subseteq N} \alpha^i_S \, v(S)$ for every $i \in N$ and $v \in \mathbb V^N$, using the notation $\alpha^i_S = \alpha^i (S)$ for every $S \subseteq N$.}
\begin{definition}
	A pair of functions $( \mu , \eta ) \colon \mathbb V^N \to \mathbb R^N \times \mathbb R^N$ is a \textbf{linear bound pair} if for every $i \in N$ there exist $\mu^i , \eta^i \in \mathbb V^N$ with $\mu_i (v) = \mu^i \cdot v$ and $\eta_i (v) = \eta^i \cdot v$ such that
	\begin{equation} \label{eq:LinearBoundDef}
		\mu^i = \sum_{j \in N} \left[ \, \sum_{T \colon j \in T} \mu^i_T \, \right] \mu^j = \sum_{j \in N} \left[ \sum_{T \colon j \in T} \eta^i_T \, \right] \mu^j
	\end{equation}
\end{definition}
The next proposition shows that the naming of a pair of linear functions as a ``bound pair'' in the definition above is justified since for every linear bound pair there is a non-empty convex cone of $\mathbb V^N$ such that this pair is indeed a bound pair.
\begin{proposition} \label{prop:LinearBoundPair}
	Every linear bound pair $(\mu , \eta )$ is a bound pair on $\mathbb V ( \mu , \eta ) = \{ v \in \mathbb V^N \mid \mu (v) \leqslant \eta (v) \} \neq \varnothing
    $, which is a non-empty closed convex cone in $\mathbb V^N$.
\end{proposition}
\begin{proof}
	Let $(\mu , \eta )$ be a linear bound pair satisfying the properties as given. \\
	We first show that $\mathbb V (\mu , \eta )$ is a non-empty closed convex cone in $\mathbb V^N$. We note that linearity of the two functions $\mu$ and $\eta$ implies that for all $v,w \in \mathbb V (\mu , \eta )$ and every player $i \in N$, there exist $\mu^ i,\ \eta^ i \in \mathbb V^N$ such that 
    \begin{equation}\label{eqmi}
    \mu_i(v)=\mu^ i \cdot v \mbox{ and } \eta_i(v)=\eta^ i \cdot v.
    \end{equation}
    Non-emptiness follows from  $v_0 =0 \in \mathbb V (\mu , \eta )$. Moreover, (\ref{eqmi}) implies that for all $\lambda_1 , \lambda_2 \geqslant 0$ and every player $i \in N \colon$
	\begin{align*}
		 \mu_i (\lambda_1 v + \lambda_2 w) & = \mu^i \cdot (\lambda_1 v + \lambda_2 w) = \lambda_1 \mu^i \cdot v + \lambda_2 \mu^i \cdot w \\
		 & = \lambda_1 \mu_i(v) + \lambda_2 \mu_i (w) \leqslant \lambda_1 \eta_i (v) + \lambda_2 \eta_i (w) = \eta_i (\lambda_1 v + \lambda_2  w) 
	\end{align*}
	where the inequality follows from $v,w \in \mathbb V (\mu , \eta )$. This shows that $\mathbb V (\mu , \eta )$ is a closed cone in $\mathbb V^N$. Convexity follows from the above by selecting $\lambda_1 + \lambda_2 =1$.   \\
    Hence, $(\mu , \eta )$ indeed satisfies Definition \ref{def:BoundPair}(i) on the non-empty closed convex cone $\mathbb V (\mu , \eta )$. \\[1ex]
	Second, we show that $(\mu , \eta )$ satisfies Definition \ref{def:BoundPair}(ii-a) on $\mathbb V^N$. Due to the linearity of $\mu$ we only have to show the desired property for every standard base game $b_S \in \mathbb V^N$, $S \in 2^N$. \\
	Let $S \subseteq N$. Then we derive for every player $i \in N$ that
	\begin{eqnarray}\label{equb}
		\mu_i (b_S - \mu (b_S) \, ) & = \mu^i \cdot b_S - \mu^i \cdot \mu (b_S) = \mu^i_S - \sum_{T \subseteq N} \mu^i_T \left( \, \sum_{j \in T} \mu^j_S \, \right) \nonumber \\
		&= \mu^i_S - \sum_{j \in N} \left[ \, \sum_{T \subseteq N \colon j \in T} \mu^i_T \, \right]  \mu^j_S = \mu^i_S - \mu^i_S =0
	\end{eqnarray}
	where $\mu^ i$ is given by (\ref{eqmi}) and the fourth equality follows from (\ref{eq:LinearBoundDef}). By writing any $v \in \mathbb V (\mu , \eta )$ as $v = \sum_{S \subseteq N} v(S) \, b_S$ we now derive that for every player $i \in N \colon$
	\[
		\mu_i (v - \mu (v)) = \mu^i \cdot \left( \sum_{S \subseteq N} v(S) \, b_S - \sum_{S \subseteq N} v(S) \, \mu(b_S) \, \right)  = \sum_{S \subseteq N} v(S) \, \mu^i \cdot (b_S - \mu (b_S) \, ) =0 
	\]
	by (\ref{equb}). Hence, property (ii-a) of Definition \ref{def:BoundPair} is satisfied on $\mathbb V^N$.
	\\[1ex]
	Third, to show that $(\mu , \eta )$ satisfies Definition \ref{def:BoundPair}(ii-b) on $\mathbb V^N$, we again can restrict ourselves to checking this property for all base games. Take $b_S \in \mathbb V^N$ for some $S \in 2^N$. Now, for every player $i \in N \colon$
	\begin{align*}
		\eta_i (b_S - \mu (b_S) \, ) & = \eta^i \cdot b_S - \eta^i \cdot \mu(b_S) =\eta^i_S - \sum_{T \subseteq N} \eta^i_T \left( \sum_{j \in T} \mu^j_S \right) \\
		& =\eta^i_S - \sum_{j \in N} \left[ \, \sum_{T \subseteq N \colon j \in T} \eta^i_T \, \right]  \mu^j_S = \eta^i_S - \mu^i_S = \eta_i (b_S) - \mu_i (b_S)
	\end{align*}
	using (\ref{eq:LinearBoundDef}). This shows property (ii-b) of Definition \ref{def:BoundPair} on $\mathbb V^N$. \\[1ex]
	Therefore, we have shown that the pair $(\mu , \eta )$ satisfies Definition \ref{def:BoundPair} and is indeed a bound pair on the non-empty closed convex cone $\mathbb V ( \mu , \eta )$.
\end{proof}

\paragraph{Bound balanced games}

For $(\mu , \eta ) \colon \mathbb V \to \mathbb R^N \times \mathbb R^N$ being a bound pair on some class of games $\mathbb V \subseteq \mathbb V^N$,  and $v \in \mathbb V$, we define
\begin{equation}
	\mathcal Q (v; \mu , \eta ) = \{ x \in \mathcal A (v) \mid \mu (v) \leqslant x \leqslant \eta (v) \, \}
\end{equation}
as the class of allocations of $v$ that are bound by $\mu (v)$ and $\eta (v)$. The following proposition clarifies further when bound pairs are meaningful in the sense that this class of allocations is nonempty.
\begin{proposition} \label{prop:BoundGame}
	Let $(\mu , \eta ) \colon \mathbb V \to \mathbb R^N \times \mathbb R^N$ be a bound pair on $\mathbb V \subseteq \mathbb V^N$ and let $v \in \mathbb V$. For the game $v$ the set of $( \mu (v) , \eta (v))$-bound allocations is non-empty if and only if the sums of the lower- and upper bound payoffs are lower- and upper bounds for $v(N)$, i.e.,
	\begin{equation} \label{eq:Q}
		\mathcal Q (v; \mu , \eta ) \neq \varnothing \quad \mbox{if and only if } \quad \sum_{i \in N} \mu_i (v) \leqslant v(N) \leqslant \sum_{i \in N} \eta_i (v) .
	\end{equation}
\end{proposition}

\begin{proof}
	Let $(\mu , \eta )$ be a bound pair on $\mathbb V \subseteq \mathbb V^N$ and let $v \in \mathbb V$.
	\\[1ex]
	{\bf (Only if)} Assume that $\mathcal Q (v; \mu , \eta ) \neq \varnothing$. Let $x \in \mathcal Q (v; \mu , \eta )$. Then $\mu (v) \leqslant x \leqslant \eta (v)$. Moreover, $x \in \mathcal A(v)$ implies that $x(N) = v(N)$. Therefore, $\sum_{i \in N} \mu_i (v) \leqslant x(N) = v(N) \leqslant \sum_{i \in N} \eta_i (v)$.
	\\[1ex]
	{\bf (If)} Assume that $\sum_{i \in N} \mu_i (v) \leqslant v(N) \leqslant \sum_{i \in N} \eta_i (v)$. By Definition \ref{def:BoundPair}(i), $\mu (v) \leqslant \eta (v)$. \\ Then it is obvious that there exists some $x \in \mathbb R^N$ such that (i) $\sum_{i \in N} x_i = v(N)$ and (ii) $\mu (v) \leqslant x \leqslant \eta (v)$. Clearly, $x \in \mathcal Q(v;\mu,\eta)$ and, thus, $\mathcal Q(v;\mu,\eta) \neq \varnothing$.
	\\[1ex] 
	Together, these show (\ref{eq:Q}).
\end{proof}

\bigskip\noindent
Proposition \ref{prop:BoundGame} allows the introduction of the subclass of cooperative games for which the set of bound allocations is non-empty for a given bound pair. 
\begin{definition}
	Let $(\mu , \eta )$ be a bound pair on $\mathbb V \subseteq \mathbb V^N$. The subclass of \textbf{$(\mu , \eta )$-balanced games} is defined by
	\begin{equation}
		\mathbb B ( \mu , \eta ) = \left\{ v \in \mathbb V \, \left| \, \sum_{i \in N} \mu_i (v) \leqslant v(N) \leqslant \sum_{i \in N} \eta_i (v) \right. \right\}
	\end{equation}
	Equivalently, by Proposition \ref{prop:BoundGame}, the class of $( \mu , \eta )$-balanced games consists of those games $v \in \mathbb V$ for which $\mathcal Q (v; \mu , \eta ) \neq \varnothing$. 
\end{definition}
The introduction of the subspace of $(\mu , \eta )$-balanced games gives rise to the question whether compromise points define a value on this class that assigns the corresponding compromise point to each game. This corresponds to finding a weight $\lambda \in [0,1]$ such that $\lambda \mu (v) + (1- \lambda ) \eta (v) \in \mathcal Q (v; \mu , \eta )$. The next proposition provides formulations for the resulting \emph{compromise value allocations}.
\begin{proposition} \label{prop:CompromisePoint}
	Let $(\mu , \eta )$ be a bound pair on $\mathbb V \subseteq \mathbb V^N$ such that the corresponding subclass of $(\mu , \eta )$-balanced games $\mathbb B ( \mu , \eta ) \subseteq \mathbb V$ is non-empty. Then for every $(\mu , \eta )$-balanced game $v \in \mathbb B ( \mu , \eta ) \colon$
	\begin{numm}
		\item If $\mu (v) < \eta (v)$, it holds that
		\begin{equation} \label{eq:CompromiseValueDef}
			\gamma (v; \mu , \eta ) = \frac{v(N) - \sum_{i \in N} \mu_i (v)}{\sum_{i \in N} ( \eta_i (v) - \mu_i (v) \, )} \, \eta (v) + \frac{\sum_{i \in N} \eta_i (v) - v(N)}{\sum_{i \in N} ( \eta_i (v) - \mu_i (v) \, )} \, \mu (v) \in \mathcal Q ( v ; \mu (v) , \eta (v) \, )
		\end{equation}
		and,
		\item if $\mu (v) = \eta (v)$, it holds that
			\begin{equation}
				\gamma (v; \mu , \eta ) = \mu (v) = \eta (v) \in \mathcal Q ( v ; \mu (v) , \eta (v) \, )
			\end{equation}
	\end{numm}
	Therefore, the map $\gamma ( \cdot ; \mu , \eta ) \colon \mathbb B ( \mu , \eta ) \to \mathbb R^N$ defines a value on $\mathbb B ( \mu , \eta )$, satisfying $\sum_{i \in N} \gamma_i (v; \mu, \eta ) = v(N)$ for every $v \in \mathbb B (\mu, \eta )$.
\end{proposition}
The proof of Proposition \ref{prop:CompromisePoint} is trivial and is, therefore, omitted. We refer to the value $\gamma ( \cdot ; \mu , \eta )$ on $\mathbb B ( \mu , \eta )$ introduced in Proposition \ref{prop:CompromisePoint} as the \emph{$(\mu , \eta )$-compromise value}.

\subsection{An axiomatic characterisation of compromise values}

We continue our discussion of compromise values by constructing an axiomatic characterisation of \emph{all} compromise values based on bound pairs as introduced in Definition \ref{def:BoundPair}.

In particular, we demonstrate that the axiomatisation of the $\tau$-value seminally developed by \citet{Tijs1981} can be extended to any arbitrary compromise value. The subsequent theorem provides a complete characterisation of compromise values in terms of their associated bound pair. This type of axiomatisation has been developed for various compromise values and classes in the literature, which is further discussed in the subsequent sections of this paper.
\begin{theorem} \label{thm:MainCharacterisation}
	Let $(\mu , \eta )$ be a bound pair on $\mathbb V \subseteq \mathbb V^N$ and let $\mathbb B ( \mu , \eta ) \subseteq \mathbb V$ be the corresponding subclass of $(\mu , \eta )$-balanced games. \\
	Then the $(\mu , \eta )$-compromise value $\gamma ( \cdot ; \mu , \eta)$ is the unique value $f \colon \mathbb B ( \mu , \eta ) \to \mathbb R^N$ which satisfies the following two properties:
	\begin{numm}
		\item \textbf{Minimal rights property:} \\ For every $v \in \mathbb B ( \mu , \eta ) \colon f(v) = f(v- \mu (v) ) + \mu (v)$.
		\item \textbf{Restricted proportionality:} \\ For every game $v \in \mathbb B (\mu , \eta)$ with $\mu (v)=0$ there exists some $\lambda_v \in \mathbb R$ such that $f(v) = \lambda_v \, \eta (v)$.
	\end{numm}
\end{theorem}

\begin{proof}
	Let $(\mu , \eta )$ be a bound pair on $\mathbb V \subseteq \mathbb V^N$ and let $\mathbb B ( \mu , \eta ) \subseteq \mathbb V$ be the corresponding subclass of $(\mu , \eta )$-balanced games.
	\\[1ex]
	We show that $\gamma ( \cdot ; \mu , \eta )$ satisfies the properties stated in the assertion. First, note that $\gamma ( \cdot ; \mu , \eta )$ is a value, since by definition it is efficient, i.e., $\sum_{i \in N} \gamma_i ( v ; \mu , \eta ) = v(N)$ for every $v \in \mathbb B ( \mu , \eta )$. \\
	Second, we show that the $(\mu , \eta )$-compromise value $\gamma ( \cdot ; \mu , \eta)$ satisfies the two properties.
	First, with reference to Proposition \ref{prop:CompromisePoint}, $\gamma ( \cdot ; \mu , \eta )$ satisfies restricted proportionality on $\mathbb B ( \mu , \eta ) \subseteq \mathbb V^N$, since  $\mu (v) =0$ for $v \in \mathbb B ( \mu , \eta )$ implies that
			\[
				\gamma (v ; \mu , \eta ) = \frac{v(N)}{\sum_{i \in N} \eta_i (v)} \, \eta (v)
			\]
		Next, to show that $\gamma ( \cdot ; \mu , \eta )$ satisfies the minimal rights property, let $v \in \mathbb B ( \mu , \eta )$. Then by Definition \ref{def:BoundPair}, $\mu (v - \mu (v))=0$ and $\eta (v- \mu (v))= \eta (v) - \mu (v)$. We distinguish two cases: \\[1ex]
		First, if $\mu (v) = \eta (v)$, then $\eta (v - \mu (v)) =0 = \mu (v - \mu (v))$. Hence, by definition of $\gamma$, for some $\lambda_1 \in [0,1] \colon$
			\[
				\gamma (v - \mu (v) ; \mu , \eta ) = \lambda_1 \eta (v- \mu (v)) + (1- \lambda_1 ) \mu (v - \mu (v)) =0, 
			\]
	and, thus, $\gamma (v; \mu , \eta ) = \mu (v) = \gamma (v - \mu (v) ; \mu , \eta ) + \mu (v)$. \\[1ex]
	Next, if $\mu (v) < \eta (v)$, we have from $\mu (v - \mu (v)) =0$ and $\eta (v - \mu (v)) = \eta (v) - \mu (v)$ by (\ref{eq:CompromiseValueDef}) that
	\begin{align*}
		\gamma (v - \mu (v) ; \mu , \eta ) & = \frac{v(N) - \sum_{i \in N} \mu_i (v) - \sum_{i \in N} \mu_i (v - \mu (v))}{\sum_{i \in N} \left( \eta_i (v - \mu(v)) - \mu_i (v- \mu(v)) \, \right)} \eta (v - \mu (v)) \\[1ex]
		& = \frac{v(N) - \sum_{i \in N} \mu_i (v)}{\sum_{i \in N} \left( \eta_i (v) - \mu_i (v)) \, \right)} \, [ \eta (v) - \mu (v) \, ]
	\end{align*}
	Hence,
	\begin{align*}
		\gamma (v - \mu (v) ; \mu , \eta ) + \mu (v) & = \frac{v(N) - \sum_{i \in N} \mu_i (v)}{\sum_{i \in N} \left( \eta_i (v) - \mu_i (v)) \, \right)} \, \eta(v) + \left[ 1- \frac{v(N) - \sum_{i \in N} \mu_i (v)}{\sum_{i \in N} \left( \eta_i (v) - \mu_i (v)) \, \right)} \, \right] \, \mu (v) \\[1ex]
		& = \frac{v(N) - \sum_{i \in N} \mu_i (v)}{\sum_{i \in N} \left( \eta_i (v) - \mu_i (v)) \, \right)} \, \eta(v) - \frac{v(N) - \sum_{i \in N} \eta_i (v)}{\sum_{i \in N} \left( \eta_i (v) - \mu_i (v)) \, \right)} \, \mu (v) \\[1ex]
		& = \gamma (v ; \mu , \eta ) .
	\end{align*}
	This shows that the compromise value $\gamma$ satisfies the minimal rights property. \\
	Next, we show that if a value $f \colon \mathbb B ( \mu , \eta ) \to \mathbb R^N$ satisfies the two stated properties, it necessarily has to be the corresponding compromise value. Combining the two properties for any $v \in \mathbb B (\mu , \eta )$ we derive that
	\[
	f(v) = f(v - \mu (v)) + \mu (v) = \lambda \, \eta (v - \mu (v)) + \mu (v) 
	\]
	for some $\lambda \in \mathbb R$ where the first equality follows from the minimal rights property and the second equality from  restricted proportionality. Using $\eta (v - \mu (v)) = \eta (v) - \mu (v)$ we subsequently conclude that $f(v) = \lambda \, \eta (v) + (1- \lambda ) \mu (v)$. \\
	If $\mu (v) = \eta (v)$, it follows immediately that $f(v) = \mu (v) = \gamma (v ; \mu , \eta )$. \\
	This leaves the case that $\mu (v) < \eta (v)$. Using the efficiency of $f (v)$, it holds that $\sum_{i \in N} f_i (v) = v(N)$. Together with $\sum_{i \in N} \mu_i (v) \leqslant v(N) = \sum_{i \in N} f_i (v) \leqslant \sum_{i \in N} \eta_i (v)$ and $\mu (v) < \eta (v)$ it follows that
	\[
		\sum_{i \in N} f_i(v) = \lambda \sum_{i \in N} \eta_i(v) + (1-\lambda)\sum_{i \in N} \mu_i(v) = v(N) 
	\]
	implying that
	\[
		\lambda \left(\sum_{i \in N} (\eta_i(v) - \mu_i(v) \right) = v(N) - \sum_{i \in N} \mu_i(v) 
	\]
	Hence,
	\[
		\lambda = \frac{v(N) - \sum_{i \in N} \mu_i (v)}{\sum_{i \in N} ( \eta_i (v) - \mu_i (v))}
	\]
	showing that $f(v) = \gamma (v ; \mu , \eta )$.
\end{proof}

\subsection{Three illustrations of compromise values}

To elucidate and refine the concepts explored in the preceding discussion, we examine three compromise values. Notably, the first two values, PANSC and Gately values, have been previously examined in the literature. In contrast, the KM-value represents a novel contribution, emerging from the integration of an upper and a lower bound, applicable over the entire space of all games $\mathbb{V}^N$.

\subsubsection{The PANSC value}

Considering the zero vector as the chosen lower bound and the marginal contributions vector as the chosen upper bound, the resulting compromise value is the \emph{Proportional Allocation of Non-Separable Contributions} (PANSC) value which has been studied extensively by \citet{PANSC2023}. 

The PANSC value assigns to every player a payoff that is proportional to the player's marginal contribution to the total wealth generated in the game. Formally, for the game $v \in \mathbb V^N$ with $\sum_{j \in N} M_j (v) \neq 0$ and player $i \in N$, the PANSC value is defined by
\begin{equation}
	\mathrm{PANSC}_i (v) = \frac{M_i (v)}{\sum_{j \in N} M_j (v)} \, v(N) .
\end{equation}
The PANSC value corresponds to the $(\mu^0 , M )$-compromise value, where $\mu^0 (v)=0$ and $M (v) = \left( M_1 (v), \ldots , M_n (v) \, \right)$ for every $v \in \mathbb B (\mu^0 ,M)$ with
\begin{equation}
	\mathbb B (\mu^0 ,M) = \left\{ v \in \mathbb V^N \, \left| \, 0 \leqslant v(N) \leqslant \sum_{j \in N} M_j (v) \, \right. \right\}
\end{equation}
We remark that this class of cooperative games includes the set of non-negative essential games. 

It is easy to show that $\mathrm{PANSC} \colon \mathbb B (\mu^0 , M ) \to \mathbb R^N$ is the unique value $f$ on $\mathbb B (\mu^0 , M )$ such that for every game $v \in \mathbb B (\mu^0 ,M)$ there exists some $ \lambda_v \geqslant 0$ with $f (v) = \lambda_v \, M (v)$.

\subsubsection{The Gately value}

Another example of a compromise value is the \emph{Gately value} initially proposed by \citet{Gately1974} and further developed by \citet{Gately1976,Gately1978,Gately2019} and \citet{GillesMallozzi2024}.

The Gately value is the compromise value based on the bound pair $(\underline{\nu} , M)$, where, for every $v \in \mathbb V^N$, $\underline{\nu} (v) = ( v_1, \ldots , v_n)$ is the vector of individual worths and $M(v) = (M_1 (v) , \ldots , M_n (v))$ is the vector of marginal contributions. The pair $(\underline{\nu} , M )$ indeed forms a bound pair\footnote{It is easy to check the three required properties of a bound pair for $(\underline{\nu} , M )$.} on the generated class of $(\underline{\nu} , M)$-balanced games $\mathbb B ( \underline{\nu} , M )$, which is exactly the class of essential games $\mathbb V^N_E$.

The corresponding $(\underline{\nu} ,M)$-compromise value on $\mathbb B ( \underline{\nu} , M ) = \mathbb V^N_E$ is the Gately value given by
\begin{equation} \label{eq:GatelyDef}
	g_i (v) = v_i + \frac{M_i (v) - v_i}{\sum_{j \in N} \left( \, M_j (v) - v_j \, \right)} \left( \, v(N) - \sum_{j \in N} v_j \, \right)
\end{equation}
for every essential game $v \in \mathbb V^N_E$ and player $i \in N$.\footnote{With reference to the discussion in Sections 4 and 5 in this paper, we remark that the Gately value for games outside the class $\mathbb B ( \underline{\nu} , M ) = \mathbb V^N_E$ cannot be constructed through the methods discussed there. Hence, the Gately value does not result from either the lower bound $\underline{\nu}$ or the upper bound $M$.}

The Gately value has some interesting properties. First, the Gately value is the value at which the so-called propensities to disrupt for all players are minimal \citep{Gately1974}. This refers to the interpretation that the Gately value is the equilibrium outcome of a bargaining process over the allocation of the generated worths in the cooperative game.

Second, the Gately value is \emph{self-dual} in the sense that the Gately value of the dual game is equal to the Gately value of the original game \citep[Proposition 3.9]{GillesMallozzi2024}. Here we remark that the dual of an essential game is essential as well. 

Third, the Gately value is in the Core of every three player cooperative game as shown by \citet[Theorem 4.2]{GillesMallozzi2024}. For games with more than three players, this might not be the case, showing that in general the relationship between compromise values and the Core is undetermined.

\subsubsection{The KM-value}

It is an interesting question whether a compromise value can be constructed on the whole class of cooperative games $\mathbb V^N$. This has been investigated by \citet{Brink1994a}, who constructed a value based on a lower bound introduced by \citet{Kikuta1980} and an upper bound that was already considered by \citet{Milnor1952}.  Both of these bounds were originally considered for the Core only.

Formally, the Kikuta lower bound $\underline{M} \colon \mathbb V^N \to \mathbb R^N$ assigns to every player $i \in N$ her minimal marginal contribution in a game $v \in \mathbb V^N$, defined as
\begin{equation}
	\underline{M}_i (v) = \min_{S \subseteq N \colon i \in S} \left( \, v(S) - v(S-i) \, \right)
\end{equation}
Similarly, the Milnor upper bound $\overline{M} \colon \mathbb V^N \to \mathbb R^N$ assigns to every player $i \in N$ his maximal marginal contribution in a game $v \in \mathbb V^N$, defined as
\begin{equation}\label{equm}
	\overline{M}_i (v) = \max_{S \subseteq N \colon i \in S} \left( \, v(S) - v(S-i) \, \right)
\end{equation}
The next proposition summarises the properties of these bounds and introduces the \emph{KM-value} as the $(\underline{M} , \overline{M} )$-compromise value.
\begin{proposition}
	The pair $(\underline{M} , \overline{M} )$ forms a bound pair on $\mathbb B (\underline{M} , \overline{M} ) = \mathbb V^N$.
\end{proposition}
\begin{proof}
	To show that $(\underline{M} , \overline{M} )$ is a bound pair on $\mathbb V^N$, we show the conditions (i), (ii-a) and (ii-b) of Definition \ref{def:BoundPair} over the whole game space $\mathbb V^N$. 
	\begin{numm}
		\item It is immediately clear that for every $v \in \mathbb V^N \colon \underline{M} (v) \leqslant \overline{M} (v)$.
		\item Next, let $v \in \mathbb V^N$ and let $v' = v - \underline{M} (v)$. To check that $\underline{M} ( v - \underline{M} (v) \, ) = \underline{M} (v') =0$ for $i \in N$, let $S_i \in 2^N$ be such that  $i \in S_i$ and $\underline{M}_i (v) = v(S_i) - v(S_i - i)$. Then for an arbitrary coalition $S \in 2^N$ we derive that
	\begin{equation}
		v' (S) = (v - \underline{M} (v) \, ) (S) = v(S) - \sum_{j \in S} \underline{M}_j (v)  = v(S) - \sum_{j \in S} v(S_j) + \sum_{j \in S} v(S_j -j ) .
	\end{equation}
	Now for any player $i \in S$ it is easy to see that
	\begin{equation} \label{eq:KM-proof}
		 v' (S) - v' (S-i) = v(S) - v(S-i) - \underline{M}_i (v) .
	\end{equation}
	This, in turn, implies that
	\begin{align*}
		\underline{M}_i(v') & = \min_{S \in 2^N: i \in S} [v'(S) - v'(S - i)] \\
		& = \min_{S \in 2^N: i \in S} (v(S) - v(S - i) \, ) -  \underline{M}_i(v) \\
		& = \underline{M}_i(v) - \underline{M}_i(v) = 0 
	\end{align*}
	where the first equality follows by definition and the second equality follows from (\ref{eq:KM-proof}). Thus, (ii) is satisfied.
	\end{numm}
	It is easily established that for every $v \in \mathbb V^N \colon \overline{M} (v - \underline{M} (v) \, ) = \overline{M} (v) - \underline{M}(v)$. Hence, with (i) and (ii) above, $(\underline{M} , \overline{M} )$ is shown to be a bound pair on $\mathbb V^N$. \\[1ex]
	To establish that $\mathbb B (\underline{M} , \overline{M} ) = \mathbb V^N$ we introduce $N_0 = \varnothing$ and for every $k = 1, \ldots , n \colon N_k = \{ 1, \ldots , k \} \subseteq N$. In particular, $N_n = N$. Hence, for every game $v \in \mathbb V^N$ and every player $i \in N$ we have that $\underline{M}_i (v) \leqslant v(N_i) - v(N_{i-1}) \leqslant \overline{M}_i (v)$. This implies that
	\[
	\sum_{i \in N} \underline{M}_i (v) \leqslant \sum_{k \in N} \left( v(N_k) - v(N_{k-1}) \, \right) = v(N_n) = v(N) \leqslant \sum_{i \in N} \overline{M}_i (v) .
	\]
	This shows that indeed $\mathbb B (\underline{M} , \overline{M} ) = \mathbb V^N$.
\end{proof}

\bigskip\noindent
We refer to the $(\underline{M} , \overline{M} )$-compromise value $\kappa \colon \mathbb V^N \to \mathbb R^N$ as the \emph{KM-value}. It is a compromise value that is defined for \emph{all} cooperative games.

We note that the KM-value $\kappa$ is also self-dual in the sense that $\kappa (v) = \kappa (v^*)$ for all $v \in \mathbb V^N$. Indeed, we note that for $v \in \mathbb V^N$ and $i \in N \colon$
\begin{align*}
	\underline{M}_i (v^*) & = \min_{S \colon i \in S} \left( v^* (S) - v^* (S-i) \, \right) = \min_{S \colon i \in S} \left( v(N) - v(N \setminus S) - v(N) + v (N \setminus (S-i) \, ) \, \right) \\
	& = \min_{S \colon i \in S} \left( v( (N \setminus S) +i) - v(N \setminus S) \, \right) =  \min_{T \colon i \in T} \left( v(T) - v(T-i) \, \right) = \underline{M}_i (v) .
\end{align*}
Similarly, we can show that $\overline{M} (v^*) = \overline{M} (v)$, and thus $\kappa (v^*) = \kappa (v)$.

\paragraph{The KM-value for convex games}

With reference to the discussion of the $\tau$-value in Section 2, we remark that all convex games are semi-balanced, i.e., $\mathbb V_C^N \subset \mathbb V^N_S$, as pointed out by \citet{Tijs1981}. The next proposition shows that for convex games, the $\tau$-value is equal to the KM-value.
\begin{proposition} \label{prop:Tau=KM}
	For every convex game $v \in \mathbb V^N_C \colon \kappa (v) = \tau (v)$. 
\end{proposition}
\begin{proof}
	Let $v \in \mathbb V_C^N$ be a convex game and let $i \in S \subset T$. Then by convexity and using the fact that $m_i (v) = v_i$ for every convex game as shown by \citet{TijsDriessen1985}, it follows immediately that for any player $i \in N \colon$
	\[
	\underline{M}_i (v) = \min_{S \colon i \in S} ( v(S) - v(S-i) \, ) = v_i = m_i (v),
	\]	
	and
	\[
	\overline{M}_i (v) = \max_{S \colon i \in S} ( v(S) - v(S-i) \, ) = v(N) - v(N-i) = M_i (v) .
	\]	
	Therefore, $\kappa (v) = \gamma (v; \underline{M} , \overline{M} ) = \gamma (v ; m,M) = \tau (v)$.
\end{proof}

\medskip\noindent
We have introduced compromise values for arbitrary bound pairs. In the next two sections, we discuss construction methods where we fix a specific natural lower, respectively upper, bound that has been considered in the literature on cooperative games, and study the associated compromise value that then only depends on the chosen upper, respectively lower, bound.

\section{Constructing compromise values from lower bounds}

In this section we examine a specific subset of compromise values derived solely from an imposed lower bound. We do this by introducing a specific procedure how to associate to every lower bound function an appropriate upper bound function. While this construction is not unique, the methodology employed here appears particularly intuitive.

\subsection{Regular lower bounds and LBC values}

The next definition introduces a category of ``regular'' lower bounds that satisfy condition (ii-a) of Definition \ref{def:BoundPair}, from which compromise values can be constructed in a relatively straightforward manner.
\begin{definition}
	The function $\mu \colon \mathbb V^N \to \mathbb R^N$ is referred to as a \textbf{regular lower bound} if for every $v \in \mathbb B_\ell ( \mu )$ it holds that $\mu (v- \mu(v)) =0$, where
	\begin{equation}
		\mathbb B_\ell (\mu) : = \left\{ v \in \mathbb V^N \, \left| \, \sum_{i \in N} \mu_i (v) \leqslant v(N) \, \right. \right\}
	\end{equation}
	is the subclass of $\mu$-lower bound games.
\end{definition}

 \noindent
The next proposition links the regularity of a lower bound to the ability to construct a natural upper bound with this lower bound such that the resulting pair forms a bound pair. This gives rise to the identification of a natural compromise value for any given regular lower bound.
\begin{proposition} \label{prop:LowerBoundValue}
	Let $\mu \colon \mathbb V^N \to \mathbb R^N$ be a regular lower bound on $\mathbb B_\ell ( \mu )$.
	\begin{abet}
		\item The function $\eta^\mu \colon \mathbb B_\ell ( \mu ) \to \mathbb R^N$ defined by
		\begin{equation} \label{eq:NaturalUpper}
			\eta^\mu_i (v) = v(N) - \sum_{j \neq i } \mu_j (v)
		\end{equation}
		forms a bound pair with $\mu$ on $\mathbb B_\ell ( \mu )$ in the sense that $\mu$ is the lower bound and $\eta^\mu$ is the upper bound over the corresponding class of $( \mu , \eta^\mu )$-balanced games $\mathbb B ( \mu , \eta^\mu ) = \mathbb B_\ell (\mu) \subset \mathbb V^N$.
		\item The corresponding $( \mu , \eta^\mu )$-compromise value on $\mathbb B_\ell (\mu)$ is given by
		\begin{equation} \label{eq:LowerBCompromiseValue}
			\gamma_i (v ; \mu ) = \mu_i (v) + \frac{1}{n} \, \left[ v(N) - \sum_{j \in N} \mu_j (v) \, \right]
		\end{equation}
		The value $\gamma ( \cdot ; \mu )$ can be denoted as the $\mu$-\textbf{Lower Bound Compromise} ($\mu$-LBC) value on $\mathbb B_\ell (\mu)$.
	\end{abet}
\end{proposition}

\begin{proof}
	Let $\mu \colon \mathbb V^N \to \mathbb R^N$ be a regular lower bound, i.e., $\mu (v- \mu (v))=0$. \\[1ex]
	(a) We show that the upper bound $\eta^\mu \colon \mathbb B_\ell(\mu) \rightarrow \mathbb R^N$ defined by (\ref{eq:NaturalUpper}) satisfies, in conjunction with $\mu$ as the lower bound, the conditions of Definition \ref{def:BoundPair}: \\ 
	(i) It is obvious that $\mu(v) \leqslant \eta^\mu (v)$ for $v \in \mathbb B _\ell ( \mu )$. \\
	(ii-a) Condition \ref{def:BoundPair}(ii-a) follows from the assumption that $\mu$ is a regular lower bound. 
	\\ 
	(ii-b) To show condition \ref{def:BoundPair}(ii-b), note that for $v \in \mathbb B_\ell (\mu ) \colon$
	\begin{align*}
		\eta^\mu_i (v - \mu (v)) & = (v- \mu (v)) (N) - \sum_{j \neq i} \mu_j (v- \mu (v)) \\
		& = (v- \mu (v)) (N) = v(N) - \sum_{j \in N} \mu_j (v) = \eta^\mu_i (v) - \mu_i (v) ,
	\end{align*}
	where the second equality follows from Definition \ref{def:BoundPair}(ii-a) which we showed above. Therefore, $(\mu , \eta^\mu )$, indeed, forms a bound pair as defined in Definition \ref{def:BoundPair}. 
	\\[1ex]
	To show that $\mathbb B(\mu,\eta^\mu ) = \mathbb B_\ell (\mu) $, we note that for every $v \in \mathbb B_\ell (\mu ) \colon$
	\begin{align*}
		\sum_{i \in N} \eta^\mu_i (v) & = \sum_{i \in N} \left[ \, v(N) - \sum_{j \neq i} \mu_j (v) \, \right] = n \, v(N) - (n-1) \, \sum_{j \in N} \mu_j (v) \\
		& = v(N) + (n-1) \left[ \, v(N) - \sum_{j \in N} \mu_j (v) \, \right] \geqslant v(N) ,
	\end{align*}
	showing that $\mathbb B ( \mu , \eta^\mu ) = \mathbb B_\ell ( \mu )$ is the corresponding class of $(\mu , \eta^\mu )$-balanced games. 
	\\[1ex] 
	(b) To show that $\gamma ( \cdot ; \mu )$ defined by (\ref{eq:LowerBCompromiseValue}) is the corresponding $(\mu,\eta^\mu )$-compromise value on $\mathbb B_\ell (\mu)$, note that for every $v \in \mathbb B_\ell (\mu ) \colon$
	\begin{align*}
		v(N) - \sum_{j \in N} \mu_j (v) & = \eta^\mu_i (v) - \mu_i (v) & \mbox{for any } i \in N , \mbox{ and} \\
		\sum_{j \in N} \eta^\mu_j (v)  - v(N) & = (n-1) \, \left[ \, v(N) - \sum_{j \in N} \mu_j (v) \, \right] .
	\end{align*}
	If $\mu (v) = \eta^\mu (v)$, by definition, $v(N) = \sum_{i \in N} \mu_i (v)$. Hence, the corresponding $(\mu , \eta^\mu )$-compromise value is trivially $\gamma (v ; \mu ) = \mu (v)$, confirming (\ref{eq:LowerBCompromiseValue}) for this case. \\
	For $\mu (v) < \eta^\mu (v)$, with the above and (\ref{eq:CompromiseValueDef}), we now compute that the $( \mu , \eta^\mu )$-compromise value $\gamma (v)=\gamma (v ;\mu,\eta)$ on $\mathbb B ( \mu , \eta^\mu ) = \mathbb B_\ell ( \mu )$ is given by
	\begin{align*}
		\gamma_i (v) & =  \frac{v(N) - \sum_{j \in N} \mu_j (v)}{\sum_{j \in N} ( \eta^\mu_j (v) - \mu_j (v) \, )} \, \eta^\mu_i (v) + \frac{\sum_{j \in N} \eta^\mu_j (v) - v(N)}{\sum_{j \in N} ( \eta^\mu_j (v) - \mu_j (v) \, )} \, \mu_i (v) \\
		& = \frac{v(N)-\sum_{j \in N} \mu_j(v)}{n(v(N)-\sum_{j \in N} \mu_j(v))} \left( v(N)-\sum_{j \neq i} \mu_j(v) \, \right) + \frac{(n-1)(v(N)-\sum_{j \in N} \mu_j(v))}{n(v(N)-\sum_{j \in N} \mu_j(v))} \mu_i(v) \\
		& = \tfrac{1}{n} \left[ \, v(N) - \sum_{j \neq i} \mu_j (v) \, \right] + \tfrac{n-1}{n} \, \mu_i (v) \\
		& = \tfrac{1}{n} \left[ \, v(N) - \sum_{j \in N} \mu_j (v) \, \right] + \mu_i (v) = \gamma_i ( v ; \mu )
	\end{align*}
	for any $i \in N$. \\ This completes the proof of the proposition.
\end{proof}

\medskip\noindent
The construction of $\eta^\mu$ in Proposition \ref{prop:LowerBoundValue} mirrors the reasoning behind the $\tau$-value: instead of subtracting the other players' marginal contributions from $v(N)$, we subtract their lower bounds $\mu_j (v)$.

\paragraph{A characterisation of LBC values}

The axiomatisation of arbitrary compromise values---developed in Theorem \ref{thm:MainCharacterisation}---can be sharpened for the subclass of compromise values that can be constructed from a lower bound functional. The next theorem states that the restricted proportionality property of Theorem \ref{thm:MainCharacterisation} can be replaced by an egalitarian division property.

\begin{theorem} \label{thm:LowerCharacterisation}
	Let $\mu \colon \mathbb V^N \to \mathbb R^N$ be a regular lower bound on the class of $\mu$-bound games $\mathbb B_\ell ( \mu )$. Then the $\mu$-LBC value is the unique value $f \colon \mathbb B_\ell (\mu ) \to \mathbb R^N$ that satisfies the following two properties:
	\begin{numm}
		\item \textbf{Minimal rights property:} \\ For every $v \in \mathbb B_\ell ( \mu ) \colon f(v) = f(v- \mu (v) ) + \mu (v)$.
		\item \textbf{Egalitarian division property:} \\ For every game $v$ with $\mu(v)=0$, there exists some $\lambda_v \in \mathbb R$ such that $f(v) = \lambda_v \, e$, where $e = (1 , \ldots , 1)$.
	\end{numm}
\end{theorem}

\begin{proof}
	It is obvious that any $\mu$-LBC value indeed satisfies these two listed properties. \\ 
	The reverse---that, if $f$ satisfies the two listed properties, it is an LBC value---is a corollary of Theorem \ref{thm:MainCharacterisation} by noticing that $\eta^\mu_i(v)=v(N)$ for all $i \in N$ if $\mu(v)=0$. \\
	 	This shows the assertion.
\end{proof}

\subsection{Two well-known LBC values}

We show that the well-known Egalitarian Value and the CIS value are both compromise values that are based on the construction method set out in Proposition \ref{prop:LowerBoundValue}, based on a well-defined regular lower bound only.

\subsubsection{The Egalitarian Value}

Taking the trivial lower bound $\mu^0 \colon \mathbb V^N \to \mathbb R^N$ with $\mu^0 (v) =0$ for every $v \in \mathbb V^N$, any function $\eta \colon \mathbb V^N \to \mathbb R^N$ with $\eta (v) \geqslant 0$ for $v \in \mathbb V \subset \mathbb V^N$ is a non-trivial upper bound for $\mu^0$ to form a bound pair, where $\mathbb V \subset \mathbb V^N$ is a subclass of games $v$ for which $\mathcal A (v)$ has a non-empty relative interior.

In particular, the upper bound $\eta^{\mu^0}=\eta^0$ constructed in Proposition \ref{prop:LowerBoundValue}(a) is given by $\eta^0 (v) = ( \, v(N), \ldots , v(N) \, ) \in \mathbb R^N$. This trivial upper bound assigns the total wealth generated in the corresponding cooperative game $v \in \mathbb V^N$ to every player in the game as an upper bound on their payoff. The class of $(\mu^0 , \eta^0 )$-balanced games is given by
\begin{equation}
	\mathbb B (\mu^0 , \eta^0 ) = \mathbb B_\ell ( \mu^0 ) = \{ v \in \mathbb V^N \mid v(N) \geqslant 0 \} .
\end{equation}
The corresponding $(\mu^0 , \eta^0 )$-compromise value given by Proposition \ref{prop:LowerBoundValue}(b) is the \emph{Egalitarian Value} $E \colon \mathbb B_\ell (\mu^0) \to \mathbb R^N$ defined by\footnote{Axiomatisations of this Egalitarian Value using axioms similar as those for the Shapley value, are given in \citet{Brink2007}. }
\begin{equation}
	E_i (v) = \frac{v(N)}{n} \qquad \mbox{for every } i \in N .
\end{equation}
The Egalitarian Value $E$ is the unique value on $\mathbb B_\ell (\mu^0 )$ such that there exists some $ \lambda \geqslant 0$ with $E (v) = \lambda \, \eta^0 (v)$, i.e., $E_i (v) = \lambda \, v(N)$. Clearly, $\lambda = \tfrac{1}{n}$.

\subsubsection{The CIS-value}

Another lower bound that is considered widely in the literature on cooperative games is that of the vector of the individual worths in a game. It forms a natural lower bound on allocated payoffs and many values considered in the literature indeed have this vector as a lower bound on the assigned payoffs. 

Formally, for any cooperative game $v \in \mathbb V^N$ this lower bound is described by the vector $\underline{\nu} (v) = ( v_1, \ldots , v_n ) \in \mathbb R^N$. Since $\underline{\nu}_i(v - \underline{\nu} (v) \, ) = v_i - v_i = 0$ for any $i \in N$, the natural lower bound $\underline{\nu}$ is regular. The corresponding class of $\underline{\nu}$-lower bound games is now identified as
\begin{equation} \label{eq:CIS-class}
	\mathbb B_\ell (\underline{\nu} ) = \left\{ v \in \mathbb V^N \, \left| \, \sum_{i \in N} v_i \leqslant v(N) \right. \right\}
\end{equation}
which includes the class of essential games. We note that $\mathbb B_\ell ( \underline{\nu} )$ is the class of games that admit a non-empty set of imputations.

Using Proposition \ref{prop:LowerBoundValue}(a), we can construct the corresponding upper bound $\eta' \colon \mathbb V^N \to \mathbb R^N$ which for every $i \in N$ is defined by
\begin{equation}
	\eta'_i (v) = v(N) - \sum_{j \neq i} v_j
\end{equation}
The resulting $(\underline{\nu} , \eta' )$-compromise value as constructed in Proposition \ref{prop:LowerBoundValue}(b) is the \emph{Centre-of-gravity of the Imputation Set} (CIS) considered by \citet{DriessenFunaki1991}, which on the class of $\underline{\nu}$-lower bound games $\mathbb B_\ell ( \underline{\nu} )$ for every $i \in N$ is defined by
\begin{equation}
	\mathrm{CIS}_i (v) = v_i + \tfrac{1}{n} \left( v(N) - \sum_{j \in N} v_j \, \right)
\end{equation}
It can easily be verified that the CIS-value is indeed equal to the $(\underline{\nu} , \eta' )$-compromise value as already remarked by \citet{Brink1994a}.\footnote{We refer also to \citet{DriessenFunaki1991,EANSC1996,BrinkFunaki2009,EANSC2019} and \citet{Sharing2022} for discussions of the CIS-value and related concepts from different perspectives.}

\medskip\noindent
A one-parameter interpolation between these two LBC values, the \emph{weighted CIS value}, has recently been axiomatised by \citet{Shan2026}; it coincides with the LBC value for the lower bound $\mu_i (v) = \alpha_i \, v_i$, which is a regular lower bound precisely at the endpoints $\alpha \in \{ 0,1 \}$.

\section{Constructing compromise values from upper bounds}

In Proposition \ref{prop:LowerBoundValue}, we provide a method to derive an associated upper bound from any regular lower bound using Equation (\ref{eq:NaturalUpper}). This also introduces a method to construct a corresponding compromise value. In this section, we explore the possibility of constructing compromise values by describing a method to associate a lower bound to an upper bound. For a given covariant upper bound $\eta$, we construct a proper corresponding lower bound $\mu^\eta$ such that $(\mu^\eta, \eta)$ forms a proper bound pair. 

Our construction method is based on the methodology developed by \citet{Tijs1981} for calculating the $\tau$-value. \citet{Chi1996} further developed this approach by combining the Milnor upper bound $\overline M$ with the Tijs lower bound construction method. \citet{Compromise2000note} extended this research and introduced a general construction method based on this methodology. Here, we further generalise this method.

\subsection{Covariant upper bounds and UBC values}

We note that the methodology as set out by \citet{Compromise2000note}, founded on Tijs's construction, is rather restrictive. Indeed, since the conditions on the class of upper-bound games are rather strict, it might be that for certain upper bounds this class is empty. 

Here, we explore a more general approach, allowing a larger class of upper-bound games defined in (\ref{eq:UboundGames}) below. It is based on a covariance condition on the selected upper bound which allows the construction of a corresponding lower bound to form a bound pair. This, in turn, allows the construction of a proper compromise value. This method is fully stated in Proposition \ref{prop:UpperBoundValue}, which extends the insight of \citet[Proposition 3.2]{Compromise2000note}.

We recall that a function $f \colon \mathbb V^N \to \mathbb R^N$ is \emph{translation covariant} if $f(v+x) = f(v) +x$ for any $x \in \mathbb R^N$. Our method constructs a compromise value from a given translation covariant upper bound.

For any translation covariant $\eta \colon \mathbb V^N \to \mathbb R^N$, define
\begin{equation} \label{eq:UboundGames}
	\mathbb B_u (\eta ) = \left\{ \, v \in \mathbb V^N \, \left| \, v(S) \leqslant \sum_{i \in S} \eta_i (v) \mbox{ for every } S \subseteq N \, \right. \right\}
\end{equation}
as the subclass of \emph{strongly $\eta$-bound games}.
\begin{proposition} \label{prop:UpperBoundValue}
	Let $\eta \colon \mathbb V^N \to \mathbb R^N$ be translation covariant on $\mathbb V^N$ and $\mathbb B_u (\eta )$ the corresponding class of strongly $\eta$-bound games. Then the function $\mu^\eta \colon \mathbb B_u ( \eta ) \to \mathbb R^N$ defined by
	\begin{equation} \label{eq:NaturalLower}
		\mu^\eta_i (v) = \max_{S \subseteq N \colon i \in S} R_i (S,v) \qquad \mbox{where } R_i (S,v) = v(S) - \sum_{j \in S -i} \eta_j (v)
	\end{equation}
	forms a bound pair with $\mu^\eta$ being the lower bound and $\eta$ being the upper bound over the corresponding class $\mathbb B_u (\eta) \subset \mathbb V^N$.
\end{proposition}

\begin{proof}
	Let $\eta \colon \mathbb V^N \to \mathbb R^N$ and $\mathbb B_u ( \eta )$ be defined as formulated in the assertion. If $\mathbb B_u ( \eta ) = \varnothing$ then the assertion of Proposition \ref{prop:UpperBoundValue} holds trivially. \\ 
	Assuming $\mathbb B_u ( \eta ) \neq \varnothing$, take any $v \in \mathbb B_u ( \eta )$. Note that by definition of $\mathbb B_u (\eta )$, $v(N) \leqslant \sum_{j \in N} \eta_j (v)$. Next, let $\mu^\eta (v)$ be as defined in (\ref{eq:NaturalLower}). \\
	First, note that $\mu^\eta_i (v) \geqslant v_i$ for every $v \in \mathbb B_u ( \eta )$ and player $i \in N$ since $R_i(i,v)=v_i$. \\
	To show that $\mu^\eta (v) \leqslant \eta (v)$, assume to the contrary that there is some player $i \in N$ with $\mu^\eta_i (v) > \eta_i (v)$. From  (\ref{eq:NaturalLower}), there exists some coalition $S_i \subseteq N$ with $i \in S_i$ and
	\[
	\mu^\eta_i (v) = R_i (S_i, v) = v(S_i) - \sum_{j \in S_i - i} \eta_j (v) > \eta_i (v) .
	\]
	But then it follows that $\sum_{j \in S_i} \eta_j (v) < v(S_i)$, which contradicts that $v \in \mathbb B_u ( \eta )$. Hence, $\mu^\eta (v) \leqslant \eta (v)$, showing that condition (i) of Definition \ref{def:BoundPair} is satisfied.
	\\[1ex]
	From \citet[Proposition 3.2]{Compromise2000note}, it immediately follows that $\mu^\eta$ defined in (\ref{eq:NaturalLower}) is translation covariant on $\mathbb V^N$. Hence, in particular this implies that $\mu^\eta (v- \mu^\eta (v)) = 0$, showing that condition (ii-a) of Definition \ref{def:BoundPair} is satisfied.
	\\[1ex]
	Finally, it follows immediately from covariance of $\eta$ and the definition of $\mu^\eta$ that condition (ii-b) of Definition \ref{def:BoundPair} is satisfied. Therefore, $(\mu^\eta , \eta )$ forms a proper bound pair on the subclass of upper bound games $\mathbb B_u ( \eta ) \subset \mathbb V^N$.
\end{proof}

\bigskip\noindent
Based on the proposition we introduce the proper class of $(\mu^\eta , \eta )$-bounded games by
\begin{equation}
	\overline{\mathbb B}_u ( \eta ) = \left\{ v \in  \mathbb B_u ( \eta ) \, \left| \, \sum_{i \in N} \mu^\eta_i (v) \leqslant v(N) \leqslant \sum_{i \in N} \eta_i (v) \, \right. \right\}
\end{equation}
Proposition \ref{prop:UpperBoundValue} introduces implicitly a compromise value that is founded on a translation covariant upper bound $\eta$. We refer to the corresponding $(\mu^\eta , \eta )$-compromise value on $\overline{\mathbb B}_u ( \eta ) \subseteq \mathbb B_u (\eta)$ as the {\em $\eta$-Upper Bound Compromise} ($\eta$-UBC) value. We note the difference with Proposition \ref{prop:LowerBoundValue} where an LBC value is defined on $\mathbb B_\ell(\mu)$, while a UBC value is defined on a possibly strict subset $\overline{\mathbb B}_u ( \eta ) \subseteq \mathbb B_u (\eta) \subset \mathbb V^N$.

\paragraph{A characterisation of UBC values.}

\citet{Compromise2000note} provided a characterisation of the compromise value constructed from a translation covariant upper bound. For completeness, we provide this characterisation here as a restatement.
\begin{lemma} \label{thm:SorianoCharacterisation}  \textbf{\emph{\citep[Theorem 3.5]{Compromise2000note}}} \\
	Let $\eta \colon \mathbb V^N \to \mathbb R^N$ be a translation covariant upper bound on $\mathbb V^N$ and let $\mu^\eta$ be the constructed lower bound from (\ref{eq:NaturalLower}). Then the $\eta$-UBC value is the unique value $f \colon \overline{\mathbb B}_u ( \eta ) \to \mathbb R^N$ that satisfies the following two properties:
	\begin{numm}
		\item \textbf{Covariance property:} \\ For every $v \in \overline{\mathbb B}_u ( \eta )$ and every $x \in \mathbb R^N \colon f(\lambda v + x) = \lambda f(v ) + x$ for any $\lambda >0$
		\item \textbf{Restricted proportionality:} \\ For every game $v \in \overline{\mathbb B}_u ( \eta )$ with $\mu^\eta (v) =0$, there is some $\lambda_v \in \mathbb R$ such that $f(v) = \lambda_v \, \eta (v)$.
	\end{numm}
\end{lemma}
Note that the characterisation given in Lemma \ref{thm:SorianoCharacterisation} deviates from the general axiomatisation of compromise values (Theorem \ref{thm:MainCharacterisation}), which is based on replacing a minimal rights hypothesis with a covariance property. 

\subsection{Two examples of UBC values}

The $\tau$-value, already discussed in Section 2 of this paper, is the quintessential example of a UBC value. The construction of the $\tau$-value shows that it is a UBC value founded on the marginal contributions vector $M$ of the players to the grand coalition $N$. Hence, the $\tau$ value is the $M$-UBC value. This is studied extensively in \citet{Tijs1981,Tijs1987,TijsDriessen1987,Tau2003} and \citet{Yanovskaya2010}. Furthermore, we note that $M$ is also an upper bound in the definitions of the PANSC and Gately values in Section 3, which consequently are \emph{not} UBC values.

 Here, we discuss shortly two further well-known compromise values that can be constructed from a well-defined upper bound. In particular, we explore two natural upper bounds that satisfy the conditions imposed in Proposition \ref{prop:UpperBoundValue}.
 
 The first upper bound considered is the Milnor bound $\overline{M}$. Following the construction method of Proposition \ref{prop:UpperBoundValue}, this results in the $\chi$-value introduced by \citet{Chi1996}, which is defined for all weakly essential games. 

The second upper bound considered concerns the residual that remains for a player $i \in N$ if all other players $j \neq i$ are paid their individual worth $v_j$. We show that for a certain class of games, the resulting compromise value from this natural upper bound is the CIS-value. With the results of Section 4.3.2, this leads to the insight that the CIS-value is a compromise value that can be constructed from a lower as well as an upper bound. 

\subsubsection{The $\chi$-value}

\citet{Chi1996} introduced the notion of the $\chi$-value as an explicit modification of the $\tau$-value through the selection of the Milnor upper bound $\overline M$ \citep{Milnor1952} instead of the marginal contributions vector. Recalling the definition of $\overline M$ in (\ref{equm}), we immediately conclude that $\mathbb B_u ( \overline{M} ) = \mathbb V^N$. Defining $\mu^{\overline{M}}$ through (\ref{eq:NaturalLower}) \citet{Chi1996} showed that
\[
\overline{\mathbb B}_u ( \overline{M} ) = \left\{ v \in \mathbb V^N \, \left| \, \sum_{i \in N} v_i \leqslant v(N) \, \right. \right\}
\]
is the class of \emph{weakly essential} games.  The $\chi$-value is now defined as the corresponding $\overline M$-UBC value on $\overline{\mathbb B}_u ( \overline{M} )$, given by $\chi = \gamma \left( \cdot ; \mu^{\overline{M}} , \overline{M} \, \right)$.

Noting that any convex game is always weakly essential, the following insight for convex games follows immediately from the proof of Proposition \ref{prop:Tau=KM} and the definition of the $\chi$-value.
\begin{corollary}
	For every convex game $v \in \mathbb V^N_C \colon \chi (v) = \tau (v) = \kappa (v)$.
\end{corollary}

\subsubsection{The CIS-value redux}

In Section 4.2.2, we developed the CIS-value as an LBC-compromise value from the well-defined lower bound $\underline{\nu} (v) = (v_1, \ldots , v_n)$ for all $v \in \mathbb V^N$. The resulting compromise value based on the construction method set out in Proposition \ref{prop:LowerBoundValue} was the $(\underline{\nu} , \eta' )$-compromise value, where $\eta'_i (v) = v(N) - \sum_{j \neq i} v_j$ for $v \in \mathbb V^N$ and $i \in N$. We concluded there that the $(\underline{\nu} , \eta' )$-compromise value is the CIS-value on $\mathbb B_\ell (\underline{\nu} )$ defined by (\ref{eq:CIS-class}). 

Using the construction method set out in Proposition \ref{prop:UpperBoundValue} based on the upper bound $\eta'$, we can construct the corresponding $(\underline{\nu} ,\eta^\prime )$-compromise value as the $\eta'$-UBC value, showing it is equal to the CIS-value. 

\medskip\noindent
Let the upper bound $\eta'$ be given as above. Then
\begin{align*}
	\mathbb B_u ( \eta' ) & = \left\{v \in \mathbb V^N \, \left| \, v(S) \leqslant \sum_{i \in S} \left( \, v(N) - \sum_{j \in N - i} v_j \, \right) \mbox{ for every } S \subseteq N \right. \right\} \\[1ex]
	& = \left\{ \, v \in V^N \, \left| \, v(S) \leqslant |S| v(N) - \sum_{i \in S} \sum_{j \in N - i} v_j \mbox{ for every } S \subseteq N \, \right. \right\} \\[1ex]
	& = \left\{ \, v \in V^N \, \left| \, v(S) \leqslant |S| v(N) - \sum_{i \in N \setminus S} |S| v_i - \sum_{i \in S} (|S|-1)v_i \mbox{ for every } S \subseteq N \, \right. \right\}  \\[1ex]
	& = \left\{ v \in \mathbb V^N \, \left| \, v(S) - \sum_{i \in S} v_i \leqslant |S| \, \left( \, v(N) - \sum_{j \in N} v_j \, \right) \mbox{ for every } S \subseteq N \, \right. \right\}
\end{align*}
Note that $\varnothing \neq \mathbb B_u ( \eta' ) \subset \{ v \in \mathbb V^N \mid v(N) \geqslant \sum_{i \in N} v_i \, \}$. Furthermore, $\mathbb B_u ( \eta' )$ contains all games of which the zero-normalisation is monotone.
\begin{proposition}
	Consider the class of cooperative games given by
	\begin{equation}
		\widehat{\mathbb B} = \left\{ v \in \mathbb V^N \, \left| \, v(S) - \sum_{i \in S} v_i \leqslant (|S| -1) \, \left( \, v(N) - \sum_{j \in N} v_j \, \right) \mbox{ for every } S \subseteq N \, \right. \right\} \subset \mathbb B_u ( \eta' ) .
	\end{equation}
	Then for every $v \in \widehat{\mathbb B}$ the corresponding lower bound defined in Proposition \ref{prop:UpperBoundValue} is $\underline{\nu} (v) = (v_1, \ldots ,v_n)$. Consequently, the constructed compromise value on $\widehat{\mathbb B}$ as asserted in Proposition \ref{prop:UpperBoundValue} is the corresponding $(\underline{\nu} , \eta' )$-compromise value, being the CIS-value.
\end{proposition}
\begin{proof}
	Consider the construction of the lower bound given in (\ref{eq:NaturalLower}) for the upper bound $\eta'$. Then for $S \subseteq N$ and $i \in S \colon$
	\begin{align*}
		R_i (S,v) & = v(S) - \sum_{j \in S-i} \eta'_j (v) = v(S) - \left[ \, (|S|-1) v(N) - \sum_{j \in S-i} \sum_{h \neq j} v_h \, \right] \\[1ex]
		& = v(S) - (|S|-1) \, \left( v(N) - \sum_{j \in N} v_j \, \right) - \sum_{j \in S-i} v_j \\[1ex]
		& = \left( v(S) - \sum_{j \in S} v_j \, \right) - (|S|-1) \, \left( v(N) - \sum_{j \in N} v_j \right) + v_i
	\end{align*}
	Hence, $R_i (S,v) \leqslant v_i$ if and only if $\left( v(S) - \sum_{j \in S} v_j \, \right) \leqslant (|S|-1) \, \left( v(N) - \sum_{j \in N} v_j \right)$. By Definition of $\hat{\mathbb B}$, this implies that for all $i \in N \colon R_i (S,v) \leqslant v_i$ for all $S \subseteq N$. Therefore, the constructed lower bound is given by $\mu (v) = \underline{\nu} (v) = (v_1, \ldots ,v_n)$ for $v \in \widehat{\mathbb B}$.
\end{proof}

\medskip\noindent
The assertion shown above asserts that the constructed compromise value on the class $\widehat{\mathbb B}$ is the CIS-value. We remark that this construction method differs from the one based on the lower bound $\underline{\nu}$ as discussed in Section 3.1.2, since for games in the subclass $\mathbb B_u (\eta' ) \setminus \widehat{\mathbb B}$ the constructed UBC value differs from the CIS-value. The next example constructs such a case. 
\begin{example}
	Consider $N = \{ 1,2,3 \}$. Let $A \in \mathbb R$ and let $v \in \mathbb V^N$ be given by $v_1=v_2=1$, $v_3=2$, $v(12)=A$, $v(13) = v(23) = 6$ and $v(N)=8$. It can easily be verified that the CIS-value of this game is given by $\mathrm{CIS} (v) = \left( 2 \tfrac{1}{3} , 2 \tfrac{1}{3} , 3 \tfrac{1}{3} \right)$ for all $A \geqslant 0$. \\[1ex]
	Next, we determine that $\eta' (v) = (5,5,6)$ for all $A \geqslant 0$. This implies that $v \in \mathbb B_u ( \eta' )$ if and only if $A \leqslant 10$. Furthermore, $v \in \widehat{\mathbb B}$ if and only if $A \leqslant 6$, implying that for $A \leqslant 6$ the corresponding constructed $\eta'$-UBC value is equal to the CIS-value. \\ 
	For $6 \leqslant A \leqslant 10$ we compute that the resulting lower bound from (\ref{eq:NaturalLower}) is $\mu (v) = (A-5,A-5,2)$. Now, $\mu_1 (v) + \mu_2 (v) + \mu_3 (v) = 2A-8 \leqslant v(N)=8$ if and only if $6 \leqslant A \leqslant 8$. In that case the resulting compromise value is the feasible balance between $\mu (v) = (A-5,A-5,2)$ and $\eta' (v) = (5,5,6)$ computed as $\gamma = \left( \, \tfrac{20-A}{12-A} \, , \, \tfrac{20-A}{12-A} \, , \, \tfrac{56-6A}{12-A} \right)$. We remark that this allocation is  the CIS-value for $A=6$, as expected, while for $A=8$ the resulting $\eta'$-UBC value is given by $\gamma = (3,3,2)$.
\end{example}
Hence, we conclude from this that the CIS-value has the special property that it is a compromise value that can be constructed from a lower as well as an upper bound on these properly constructed subclasses of games. However, the difference of these construction methods implies that the UBC class $\widehat{\mathbb B}$ is smaller than the identified class for the construction of the CIS-value from the lower bound $\underline{\nu}$ in Section 4.3.2. For games outside $\widehat{\mathbb B}$, the identified $(\underline{\nu} , \eta' )$-compromise value can be different from the CIS-value and is harder to identify.

\section{Concluding remarks: The EANSC value}

The ``Egalitarian Allocation of Non-Separable Contributions'' value or EANSC value can be understood as the dual of the CIS-value. It assigns to every player her marginal contribution and then equally taxes all players for the resulting deficit.

In this section, we delve into the unique nature of the EANSC value as a compromise value for two mutually exclusive bound pairs. We demonstrate that it can be derived from the marginal contribution bound $M$ as both an upper bound and a lower bound. Consequently, the two resulting bound pairs encompass the entire space of TU-games on which the EANSC value is indeed properly defined.

\paragraph{The EANSC value as an upper bound based value}

The EANSC value has a special role in relation to UBC values. It is itself not a UBC value, but nevertheless it has an interesting relationship with the marginal contributions vector as an upper bound on its payoffs. 

The next proposition introduces an innovative perspective on the EANSC value by showing that the EANSC value is a compromise value for the marginal contribution vector $M$ as an upper bound. However, its lower bound is constructed from a different methodology as the one introduced for UBC-compromise values. This lower bound can be identified as a solution to a system of equations. Furthermore, the proposition shows that the EANSC value is defined as a compromise value on a rather large class of games, which includes all essential games.

\begin{proposition}
	Let $\widetilde{\mathbb B}  = \{ v \in \mathbb V^N \mid v(N) \leqslant \sum_{j \in N} M_j (v) \}$ be the subclass of $M$-upper bounded cooperative games, where $M_j (v) = v(N) - v (N-j)$ is the marginal contribution for player $j \in N$ in the game $v \in \mathbb V^N$. \\
	Let $\tilde \mu \colon \widetilde{\mathbb B}  \to \mathbb R^N$ be a solution to the system of $n$ equations
	\begin{equation} \label{eq:TildeMu}
		\sum_{j \neq i} \tilde\mu_j (v) = v(N-i) \qquad \mbox{for } i \in N .
	\end{equation}
	Then the EANSC value defined by
	\begin{equation} \label{eq:EANSC}
		\mathrm{EANSC}_i (v) = M_i (v) + \tfrac{1}{n} \left( v(N) - \sum_{j \in N} M_j (v) \, \right) \qquad \mbox{for every } v \in \mathbb V^N \mbox{ and } i \in N
	\end{equation}
	is the $( \tilde\mu , M)$-compromise value on $\widetilde{\mathbb B} $.
\end{proposition}

\medskip
\begin{proof}
	We construct the proof of the assertion through the method set out in Proposition \ref{prop:LowerBoundValue} based on the introduced lower bound $\tilde\mu$ defined implicitly as a solution of (\ref{eq:TildeMu}). 
	\\ 
	For every $v \in \widetilde{\mathbb B} $ and every $i \in N$, let $\tilde{\mu}(v)$ be given by
	\begin{equation}
		\tilde\mu_i (v) = M_i (v) + \tfrac{1}{n-1} \left[ \, v(N) - \sum_{j \in N} M_j (v) \, \right]
	\end{equation}
	We claim that the given $\tilde \mu$ is a solution of (\ref{eq:TildeMu}). Indeed, for every $v \in \widetilde{\mathbb B} $ and $i \in N \colon$
	\[
	\sum_{j \neq i} \tilde\mu_j (v) = \sum_{j \neq i} M_j (v) + \left( \, v(N) - \sum_{j \in N} M_j (v) \, \right) = v(N) - M_i (v) = v(N-i) .
	\]
	Next, a direct computation gives $\sum_{i \in N} \tilde\mu_i (v) = \tfrac{n}{n-1} \, v(N) - \tfrac{1}{n-1} \sum_{j \in N} M_j (v)$, so that $\sum_{i \in N} \tilde\mu_i (v) \leqslant v(N)$ if and only if $v(N) \leqslant \sum_{j \in N} M_j (v)$, i.e., $v \in \widetilde{\mathbb B} $. Hence, $\mathbb B_\ell (\tilde\mu ) = \widetilde{\mathbb B} $. \\
	Next we check that $\tilde\mu (v - \tilde\mu (v))=0$. For that we note that for every $v \in \widetilde{\mathbb B} $ and $i \in N \colon$
	\begin{align*}
		M_i (v- \tilde\mu (v)) & = M_i (v) - \tilde\mu_i (v) \quad \mbox{and} \\
		(v- \tilde\mu (v)) (N) & = v(N) - \sum_{j \in N} \tilde\mu_j (v)
	\end{align*}
	leading to the conclusion that
	\begin{align*}
		(v- \tilde\mu (v)) (N) - \sum_{j \in N} M_j (v - \tilde\mu (v)) & = v(N) - \sum_{j \in N} \tilde\mu_j (v) - \sum_{j \in N} M_j (v) + \sum_{j \in N} \tilde\mu_j (v) \\
		& = v(N) - \sum_{j \in N} M_j (v) .
	\end{align*}
	Therefore,
	\begin{align*}
		\tilde\mu_i (v- \tilde\mu (v)) & = M_i (v- \tilde\mu (v)) + \tfrac{1}{n-1} \left[ (v- \tilde\mu (v))(N) - \sum_{j \in N} M_j (v- \tilde\mu (v)) \, \right] \\
		& = M_i (v) - \tilde\mu_i (v) + \tfrac{1}{n-1} \left[ v(N) - \sum_{j \in N} M_j (v) \, \right] =\tilde\mu_i (v) - \tilde\mu_i (v) =0
	\end{align*}
	This shows that $\tilde\mu$ is indeed a regular lower bound and that we can apply the method set out in Proposition \ref{prop:LowerBoundValue}. Hence, we determine that for the formulated lower bound $\tilde\mu$ the corresponding upper bound, using (\ref{eq:NaturalUpper}) and (\ref{eq:TildeMu}), is given by
	\[
	\tilde\eta_i (v) = v(N) - \sum_{j \neq i} \tilde\mu_j (v) = v(N) - v(N-i) = M_i (v) ,
	\]
	Furthermore, the corresponding $( \tilde\mu , M)$-compromise value defined on $\mathbb B_\ell (\tilde\mu ) = \widetilde{\mathbb B} $ is for every $v \in \mathbb B_\ell ( \tilde\mu ) = \widetilde{\mathbb B} $ and $i \in N$ given by
	\begin{align*}
		\gamma_i (v; \tilde\mu ) & = \tilde\mu_i (v) + \tfrac{1}{n} \left[ \, v(N) - \sum_{j \in N} \tilde\mu_j (v) \, \right] \\[1ex]
		& = M_i (v) + \tfrac{1}{n-1} \left[ \, v(N) - \sum_{j \in N} M_j (v) \, \right] + \tfrac{1}{n} \left[ v(N) - \tfrac{1}{n-1} \sum_{j \in N} v(N-j) \, \right] \\[1ex]
		& = M_i (v) + \tfrac{1}{n-1} \left[ \, v(N) - \sum_{j \in N} M_j (v) \, \right] + \tfrac{1}{n(n-1)} \left[ \, \sum_{j \in N} M_j (v) - v(N) \, \right] \\[1ex]
		& = M_i (v) + \tfrac{1}{n} \left[ \, v(N) - \sum_{j \in N} M_j (v) \, \right] = \mathrm{EANSC}_i (v) ,
	\end{align*}
	where the second equality follows by (\ref{eq:TildeMu}) and the third equality follows from the property that the definition of $M_j(v)$ implies that $\sum_{j \in N} v(N - j) = n \, v(N) - \sum_{j \in N} M_j(v)$. \\
	This shows the assertion of the proposition.
\end{proof}

\paragraph{The EANSC value as an LBC value}

From the definition of the EANSC value it should be immediately clear that the EANSC value corresponds to the LBC value for the lower bound $M$ (Proposition \ref{prop:LowerBoundValue}).\footnote{We note that the marginal contribution function $M$ indeed forms a regular lower bound, since $M(v-M(v))=M(v) - M(v)=0$ for all games $v \in \mathbb V^N$.} The EANSC value as an LBC value is, therefore, defined over the class $\mathbb B_\ell (M) = \left\{ v \in \mathbb V^N \, \left| \, v(N) \geqslant \sum_{i \in N} M_i(v) \, \right. \right\}$. From Proposition \ref{prop:LowerBoundValue} it follows that the EANSC value is the $(M , \eta^M )$-value over $\mathbb B_\ell (M)$ for the constructed upper bound given by
\[
\eta^M_i (v) = v(N) - \sum_{j \neq i} M_j(v) = \sum_{j \neq i} v(N-j) - (n-2) v(N) .
\]
We now note that the class of the $M$-lower bounded games $\mathbb B_\ell (M)$ and the class $\widetilde{\mathbb B} $, over which the EANSC value is constructed as the $( \tilde \mu , M)$-compromise value, together cover the whole space $\mathbb V^N$. Hence, since $\mathbb B_\ell (M) \cup \widetilde{\mathbb B}  = \mathbb V^N$, combining these two characterisations of the EANSC value, we have shown that the EANSC value is a compromise value over the whole space $\mathbb V^N$ of TU-games, although for two different bound pairs.

\singlespace
\bibliographystyle{apalike}
\bibliography{RPDB}

\end{document}